\documentclass[12pt,letterpaper]{article}
\usepackage[T1]{fontenc}
\usepackage[margin=1in]{geometry}
\usepackage{natbib}
\usepackage{amsmath, amssymb, amsthm, bm, bbm}
\usepackage{enumitem}
\usepackage{hyperref}
\hypersetup{
	colorlinks=true,
	linkcolor=black,
	citecolor=black,
	urlcolor=black,
	filecolor=black
}
\usepackage{booktabs}
\usepackage{graphicx}

\newcommand{\xvec}{\bm x}
\newcommand{\yvec}{\bm y}
\newcommand{\Xmat}{\bm X}
\newcommand{\Ymat}{\bm Y}

\newcommand{\fvec}{\bm f}
\newcommand{\gvec}{\bm g}
\newcommand{\hvec}{\bm h}
\newcommand{\epvec}{\bm \varepsilon}
\newcommand{\amat}{\bm \alpha}

\newcommand{\reals}{\mathbb{R}}
\newcommand{\ind}{\perp\!\!\!\!\perp}
\newcommand{\iid}{\overset{i.i.d.}{\sim}}

\newcommand{\Eset}{\mathcal{E}}
\newcommand{\envlp}{\Eset_{\bm \Sigma} (\fvec) }
\newcommand{\Gmat}{\bm \Gamma}
\newcommand{\Omat}{\bm \Omega}
\newcommand{\Sigmat}{\bm \Sigma}

\newcommand{\hatenvlp}{\widehat{\Eset}_{\bm \Sigma}^{\lambda} (\fvec)  }
\newcommand{\hatGmat}{\widehat{\Gmat}^{\lambda}}
\newcommand{\norm}[1]{\lVert #1 \rVert}
\newcommand{\inprod}[2]{\langle #1, #2 \rangle}

\newcommand{\RKHS}{\mathcal{H}_{K}}

\DeclareMathOperator{\spann}{span}
\DeclareMathOperator{\E}{E}
\DeclareMathOperator{\Cov}{Cov}
\DeclareMathOperator{\Var}{Var}
\DeclareMathOperator{\dom}{dom}

\DeclareMathOperator*{\argmax}{arg\,max}
\DeclareMathOperator*{\argmin}{arg\,min}
\DeclareMathOperator{\tr}{tr}
\DeclareMathOperator{\diag}{diag}
\DeclareMathOperator{\vect}{vec}

\newtheorem{theorem}{Theorem}[section]
\newtheorem{prop}[theorem]{Proposition}

\theoremstyle{definition}
\newtheorem{definition}[theorem]{Definition}

\def\spacingset#1{\renewcommand{\baselinestretch}%
	{#1}\small\normalsize} 

\begin{document}
	\spacingset{1}
	
	\title{Nonparametric Envelopes for Flexible Response Reduction}
	\author{Tate Jacobson%
		\hspace{.2cm}\\
		Department of Statistics, Oregon State University}
	\maketitle

	\begin{abstract}
		Envelope methods improve the estimation efficiency in multivariate linear regression by identifying and separating the material and immaterial parts of the responses or the predictors and estimating the regression coefficients using only the material part. 
		Though envelopes have been extended to other models, such as GLMs, these extensions still largely fall under the restrictive parametric modeling framework. 
		In this paper, we introduce a flexible, nonparametric extension of response envelopes for improving efficiency in nonlinear multivariate regressions. 
		We propose the \textit{kernel envelope} (KENV) estimator for simultaneously estimating the response envelope subspace and the enveloped nonparametric conditional mean function in a reproducing kernel Hilbert space, with a novel penalty that accounts for the envelope structure. 
		We prove that the prediction risk for KENV converges to the optimal risk as the sample size diverges and show that KENV achieves a lower in-sample prediction risk than kernel ridge regression when the response has a non-trivial immaterial component. 
		We compare the prediction performance of KENV with other envelope methods and kernel regression methods in simulations and a real data example, finding that KENV delivers more accurate predictions than both the envelope-based and kernel-based alternatives in both low and high dimensions.
	\end{abstract}

	\noindent%
	\vfill
	
	\newpage

	\spacingset{1.7}
	\section{Introduction}
	Envelope methods for multivariate regression perform targeted dimension reduction, seeking to identify and separate the material and immaterial parts of the response $\yvec$ \citep{Cook2010} or the predictors $\xvec$ \citep{Cook2013} in order to improve the efficiency in estimating parameters of interest.
	Under a response envelope model, for example, there exists an immaterial part of $\yvec$ which does not depend on $\xvec$.
	By discarding the immaterial part and regressing only the material part on $\xvec$, the response envelope estimator improves the efficiency of the regression coefficient estimates. %
	As envelope methodology has been extended beyond standard linear regression---with recent extensions to generalized linear models and Cox regression \citep{Cook_foundations_2015}, reduced rank regression \citep{cook_reduced_rank_2015}, quantile regression \citep{Ding2021}, and Huber regression \citep{zhou_enveloped_2024}---improving the estimation efficiency has remained the central motivation.

	Despite this proliferation of envelope methods, few fall outside the restrictive parametric modeling framework.
	While there is a sizable literature on nonparametric and kernel-based sufficient dimension reduction (SDR) \citep{fukumizu_kernel_2009, yeh_nonlinear_2009, lee_general_2013, fukumizu_gradient-based_2014, liu_sparse_2024} and some recent work on semiparametric predictor envelopes \citep{ma_semiparametrically_2024}, this research only addresses dimension reduction for predictors, not responses. 
	In recent work which covers both predictor and response reduction, \citet{zhang_envelopes_2020} introduced the martingale difference divergence envelope and the central mean envelope, model-free extensions of envelopes which allow for heteroscedasticity and a nonlinear conditional mean function.
	Due to their model-free nature, however, these extensions of envelopes provide estimates of their targeted subspaces without corresponding estimates of the unknown conditional mean function, and their theoretical results concern the subspace estimation accuracy rather than prediction or estimation of the regression function.
	Recently, \citet{kwon_enhanced_2024} introduced the enhanced envelope estimator, a regularized version of response envelopes for multivariate regressions with high-dimensional predictors.
	While their enhanced envelope extends response envelope methodology to high dimensions, it still assumes that the underlying model is linear. 

	In this paper, our focus lies on improving efficiency in nonlinear multivariate regressions via envelopes. 
	To this end, we introduce a flexible, nonparametric extension of response envelope methodology.
	Under our nonparametric envelope model, the relationship between $\yvec$ and $\xvec$ takes the general form
	$
		\yvec = \fvec(\xvec) + \epvec
	$
	where $\fvec$ is an unknown measurable function and $\yvec$ can be decomposed into material and immaterial parts, of which only the former depends of $\xvec$.
	We propose the \textit{kernel envelope} (KENV) estimator for simultaneously estimating the response envelope and the unknown regression function $\fvec$ in a reproducing kernel Hilbert space (RKHS), featuring a novel penalty which accounts for the envelope structure of the response.
	KENV improves on existing multivariate kernel regression methods by leveraging the envelope structure of the response to increase the efficiency in estimating the regression function.
	At the same time, KENV provides a flexible nonparametric extension of envelope methodology---in addition to providing the ability to estimate complex nonlinear regression functions, our RKHS-based approach also naturally accommodates high-dimensional predictors.

	We highlight a few key contributions of this work:
	First and foremost, we introduce KENV, which is, to our knowledge, the first envelope estimator to directly and simultaneously estimate the envelope subspace and the corresponding enveloped nonparametric conditional mean function for $\yvec$.
	We describe how KENV is estimated in an RKHS and detail how we efficiently compute the solution path along a sequence of penalty parameter values.   
	We develop kernel learning theory for KENV, proving that its prediction risk converges to the optimal risk as the sample size diverges to infinity.
	This result holds under much weaker assumptions about the relationship between $\yvec$ and $\xvec$ than the strict linearity assumption of parametric envelope methods, and makes KENV a reliable option for envelope regression when linearity does not hold.
	In addition, we show that KENV achieves a lower in-sample prediction risk than multivariate kernel ridge regression if the response has a non-trivial immaterial component.
	This improvement stems from a decrease in the variance due to the response reduction, aligning with our intuition for how envelopes can enhance nonparametric regression methods.
	
	We demonstrate KENV's superior prediction performance relative to existing envelope methods and kernel-based estimators in both simulations and a real data example.
	In particular, we observe that KENV consistently delivers more accurate predictions than kernel ridge regression and support vector machines while the parametric enhanced envelope estimator \citep{kwon_enhanced_2024} can do worse than ridge regression, even when the true model is an envelope.
	This result reveals the sensitivity of a linear envelope method to model misspecification and underscores the value of the nonparametric envelope model.
	
	\subsection*{Outline}
	The rest of the paper is organized as follows.
	In Section \ref{sec:nonparametric envelopes}, we introduce the nonparametric envelope model.
	In Section \ref{sec:envelope in RKHS}, we propose the KENV estimator in RKHS and describe its computation.
	In Section \ref{sec:theory}, we study the prediction risk, comparing the in-sample risks of KENV and kernel ridge regression and developing asymptotic learning theory for KENV.
	In Section \ref{sec:sims}, we conduct simulation studies comparing the prediction accuracy of KENV with that of other envelope estimators and RKHS regression methods.
	In Section \ref{sec:real data}, we apply KENV and other methods to predict the characteristics of output from a low-density polyethylene reactor. 
	Section \ref{sec:discussion} closes with some discussion.
	Theoretical proofs are given in the appendix.
	
	\subsection*{Notation}
	We use the following notation throughout the paper.
	For any integer $n$, we define $[n] = \{1, \ldots, n\}$.
	For a matrix $\bm A \in \reals^{p \times q}$, we let $\bm A^T$ denote its transpose and $\norm{\bm A}_F = \sqrt{\sum_{i,j} a_{ij}^2}  $ denote its Frobenius norm.
	We use $\bm A_{i \cdot}$ and $\bm A_{\cdot j}$ to denote the $i$th row and $j$th column, respectively.
	For a square matrix, we use $\tr\{\bm A\}$ to denote the trace and $|\bm A|$ to denote the determinant
	For a vector $\bm v \in \reals^p$, we let $\norm{\bm v}_2 = \sqrt{ \sum_{i = 1}^p v_j^2 }$ denote the $L_2$ norm.
	We use $V_n \to_p V$ to denote convergence in probability.
	
	\section{Nonparametric Response Envelopes}\label{sec:nonparametric envelopes}
	We start with a general multivariate regression model:
	\begin{equation}
		\yvec = \fvec(\xvec) + \epvec \label{eqn:regression model}
	\end{equation}
	where $\yvec \in \reals^r, \xvec \in \reals^p$, $\fvec: \reals^p \mapsto \reals^r$ is an unknown measurable function, and $\epvec$ has mean zero and positive-definite covariance matrix $\bm \Sigma$ and is independent of $\xvec$.
	To facilitate presentation of the method we assume $\E[\yvec] = \bm 0$ without loss of generality.
	As the form of $\fvec$ is unspecified, \eqref{eqn:regression model} is a nonparametric model.

	To improve the efficiency of our estimate of $\fvec$, we seek to identify the smallest subspace $\Eset \subseteq \reals^r$ satisfying the following conditions:
	\begin{enumerate}[label = (\roman*)]
		
		\item $\bm Q_{\Eset} \yvec | \fvec(\xvec) = \bm z_1 \sim  \bm Q_{\Eset} \yvec | \fvec(\xvec) = \bm z_2$ for all $\bm z_1, \bm z_2$ in the range of $\fvec$ and \label{condition i}
		\item $\bm P_{\Eset} \yvec \ind \bm Q_{\Eset} \yvec | \fvec(\xvec)$ \label{condition ii}
	\end{enumerate}
	where $\bm P_{\Eset}$ denotes the projection onto $\Eset$ and $\bm Q_{\Eset} = \bm I - \bm P_{\Eset}$.
	Condition \ref{condition i} provides that the distribution of $\bm Q_{\Eset} \yvec$ is invariant to changes in $\fvec(\xvec)$ while \ref{condition ii} states that $\bm P_{\Eset} \yvec$ and $\bm Q_{\Eset} \yvec$ are conditionally independent given $\fvec(\xvec)$.
	Together \ref{condition i} and \ref{condition ii} imply that $\bm P_{\Eset} \yvec$ contains all of the information in $\yvec$ that can be explained by $\fvec(\xvec)$.
	For this reason, we refer to $\bm P_{\Eset} \yvec$ as the \textit{material part} of $\yvec$ for the regression and $\bm Q_{\Eset} \yvec$ as the \textit{immaterial part}.
	By separating the material and immaterial parts of $\yvec$ and performing regression on only the material part, we can increase the efficiency of our estimator of $\fvec(\xvec)$.
	
	Envelopes enable us to formalize the idea of the smallest subspace satisfying \ref{condition i} and \ref{condition ii}.
	Traditional response envelopes \citep{Cook2010}, however, assume a linear relationship between $\yvec$ and $\xvec$. %
	To accommodate our more general regression model \eqref{eqn:regression model}, we propose a novel nonparametric extension of response envelopes.
	First, we need to establish some preliminaries.
	The following proposition proposes equivalent conditions to \ref{condition i} and \ref{condition ii} which we will use to motivate our definition of envelopes:
	\begin{prop} \label{prop:condition equivalence}
		Under the nonparametric model \eqref{eqn:regression model}, \ref{condition i} holds if and only if
		\begin{enumerate}[label=(a)]
			
			\item $\fvec(\xvec) \in \Eset$ for all $\xvec \in \dom(\fvec)$.%
			\label{condition a}
		\end{enumerate}
		If, in addition, $\epvec$ is normally distributed, then \ref{condition ii} holds if and only if
		\begin{enumerate}[label=(b)]
			
			\item $\Cov(\bm P_{\Eset} \yvec, \bm Q_{\Eset} \yvec | \fvec(\xvec) ) = \bm 0 $. \label{condition b}
		\end{enumerate}
	\end{prop}
	
	Condition \ref{condition b} enables us to decompose $\bm \Sigma$ as $\bm \Sigma = \bm P_{\Eset} \bm \Sigma \bm P_{\Eset} + \bm Q_{\Eset} \bm \Sigma \bm Q_{\Eset}$.
	We generalize this property with the notion of a \textit{reducing subspace}:
	\begin{definition}
		We say that a subspace $\mathcal{R} \subseteq \reals^r$ is a \emph{reducing subspace} of a symmetric matrix $\bm M \in \reals^{r\times r}$ if it decomposes $\bm M$ as
		$
		\bm M = \bm P_{\mathcal R} \bm M \bm P_{\mathcal R} + \bm Q_{\mathcal R} \bm M \bm Q_{\mathcal R}.
		$
	\end{definition}
	\noindent Using these terms, we define the nonparametric response envelope as follows:
	\begin{definition}
		Let $\bm M \in \reals^{r \times r}$ be a positive definite matrix and $\fvec: \reals^p \mapsto \reals^r$.
		The $\bm M$-envelope of $\fvec$, denoted as $\Eset_{\bm M}(\fvec)$, is the intersection of all reducing subspaces of $\bm M$ such that  $\fvec(\xvec) \in \Eset$ for all $\xvec \in \dom(\fvec)$.
	\end{definition}
	Proposition \ref{prop:condition equivalence} provides that the $\bm \Sigma$-envelope of $\fvec$ is the smallest subspace satisfying conditions \ref{condition i} and \ref{condition ii}.
	While we have used these conditions to motivate the development of the nonparametric response envelope, the definition of the envelope does not depend on them.
	In particular, we can drop the assumption that $\epvec$ is normally distributed and still have that $\Eset_{\bm M}(\fvec)$ satisfies \ref{condition i} and \ref{condition b}, which is sufficient for our theoretical study.

	Let $u = \dim\{ \envlp \}$ denote the dimension of the envelope, $\Gmat \in \reals^{r\times u}$ be an orthogonal basis matrix for $\envlp$, and $\Gmat_0 \in \reals^{r\times (r-u)}$ be an orthogonal basis matrix for the complement of $\envlp$.
	As a consequence of condition \ref{condition a}, we know that there exists a function $\gvec: \reals^p \mapsto \reals^{u}$ such that $\fvec(\xvec) = \Gmat \gvec(\xvec)$ for all $\xvec \in \dom(\fvec)$.
	Under the nonparametric envelope model we can rewrite \eqref{eqn:regression model}  as
	\begin{equation}
		\yvec = \Gmat \gvec(\xvec)+\epvec \text{, \;\;} \bm \Sigma = \Gmat \Omat \Gmat^T + \Gmat_0 \Omat_0 \Gmat_0^T  \label{eqn:envelope model}
	\end{equation}
	where $\Omat \in \reals^{u\times u}$ and $\Omat_0 \in \reals^{(r-u)\times(r-u)}$ are positive-definite matrices.
	Under the envelope model, we see that the material part of $\yvec$ satisfies $\Gmat^T \yvec = \gvec(\xvec) + \Gmat^T \epvec $, with $\Var (\Gmat^T \epvec) = \Omat$, and the immaterial part follows the zero-mean model $\Gmat_0^T \yvec = \Gmat_0^T\epvec$, with $\Var(\Gmat_0^T\epvec) = \Omat_0$.
	If we knew a basis $\Gmat$ for $\envlp$ in advance, we could estimate $\gvec$ using only the material part $\Gmat^T \yvec$, removing the noise from the immaterial part $\Gmat_0^T \yvec$ and thereby decreasing the variance.
	In practice, however, $\envlp$ is unknown and must be estimated.
	
	\section{Estimation in a Reproducing Kernel Hilbert Space}\label{sec:envelope in RKHS}
	Suppose we observe a sample $\{(\yvec_i, \xvec_i)\}_{i=1}^n$ of i.i.d. realizations of $(\yvec, \xvec)$. Let $\Ymat = [\yvec_1 \cdots \yvec_n]^T$ and $\Xmat = [\xvec_1 \cdots \xvec_n]^T$ denote the observed response and predictor matrices.
	If $\yvec$ is normally distributed given $\xvec$ and the true envelope dimension $u$ is known, then the log-likelihood under the nonparametric response envelope model \eqref{eqn:envelope model} is given by
	\begin{align}
		\ell_u(\Eset, \gvec, \Omat, \Omat_0) = & -(nr/2) \log (2 \pi) - (n/2)\log| \Gmat \Omat \Gmat^T + \Gmat_0 \Omat_0 \Gmat_0^T | \notag\\
		& - (1/2) \sum_{i=1}^n (\yvec_i - \Gmat \gvec(\xvec_i))^T(\Gmat \Omat \Gmat^T + \Gmat_0 \Omat_0 \Gmat_0^T )^{-1} (\yvec_i - \Gmat \gvec(\xvec_i)) .\label{eqn:log likelihood}
	\end{align}
	
	In the traditional linear response envelope model, where $\gvec(\xvec) = \bm \nu \xvec$ for some $\bm \nu \in \reals^{u \times p}$, the envelope estimator is obtained by maximizing the log-likelihood.
	For the nonparametric envelope model \eqref{eqn:envelope model}, however, we need to impose some regularity on the regression function $\gvec$ to avoid obtaining an estimator that simply interpolates the training data.
	To that end, we assume that the component functions $g_j$ of $\gvec = (g_1,\ldots, g_u)^T$ belong to $\RKHS$, the reproducing kernel Hilbert space (RKHS) generated by the positive definite kernel $K(\cdot, \cdot):\mathcal X \times \mathcal{X} \mapsto \reals$,	where $\mathcal{X} \in \reals^p$ denotes the domain of $\xvec$.
	Estimation in an RKHS can capture complex nonlinear dependencies while remaining computationally tractable \citep{wahba_spline_1990}, making these methods suitable for diverse regression and classification problems (see \citet{hastie_elements_2009} for a thorough introduction).

	The choice of kernel determines what kinds of functions belong to the RKHS.
	Well-known examples of kernels include
	\begin{itemize}
		\item the Gaussian kernel: $K(\xvec, \xvec') = \exp(- \norm{\xvec - \xvec'}_2^2/\sigma^2)$
		\item the polynomial kernel: $K(\xvec, \xvec')= (\xvec^T\xvec' + 1)^\sigma$
		\item the exponential kernel: $K(\xvec, \xvec') = \exp(\xvec^T\xvec')$
	\end{itemize}
	where $\sigma$ is a kernel hyperparameter.
	For some kernels, such as the Gaussian and exponential kernels, $\RKHS$ is dense in the space of continuous functions on $\mathcal{X}$ \citep{steinwart_influence_2001}, meaning that for any continuous function $f:\mathcal{X} \to \reals$ there exists a function $h \in \RKHS$ that is arbitrarily close to $f$.
	Other kernels may be chosen based on characteristics of the data, such as the periodic kernel for estimating a periodic function.
	In addition to capturing nonlinear relationships, estimation in an RKHS naturally accommodates high-dimensional predictors, further enhancing its flexibility.

	By Mercer's theorem, the kernel admits the expansion
	$K(\xvec, \xvec') = \sum_{i=1}^{\infty} \gamma_i\phi_i(\xvec) \phi_i(\xvec')$ where $\gamma_i\geq0$ for all $i$ and $\sum_{i=1}^\infty \gamma_i^2 < \infty$.
	Any function $h \in \RKHS$ can be expressed as a sum of the eigen-functions $h(\xvec) = \sum_{i=1}^{\infty} c_i \phi_i(\xvec)$, with squared norm given by $\norm{h}_{\RKHS}^2 = \sum_{i=1}^{\infty} c_i^2/\gamma_i < \infty$.
	The inner product in $\RKHS$ is given as $\inprod{\sum_{i=1}^{\infty} c_i \phi_i(\xvec)}{ \sum_{j=1}^{\infty} d_j \phi_j(\xvec) }_{\RKHS} = \sum_{i = 1}^\infty c_i d_i/\gamma_i $.
	For $\bm h = (h_1,\ldots, h_u)^T$ with each $h_i \in \RKHS$, we define the multivariate penalty functional to be
	$
		\rho (\bm h) = \sum_{i=1}^u \norm{h_i}_{\RKHS}^2.
	$
	
	We propose a novel penalty which accounts for the envelope structure of the response, given by
	\begin{equation}
		\rho(\gvec, \Omat) =\rho(\Omat^{-1/2} \gvec).
	\end{equation}
	We see that the penalty applies regularity to $\gvec$ after standardizing the material part of the regression to have uncorrelated errors, as
	$\Omat^{-1/2}\Gmat^T \yvec = \Omat^{-1/2}\gvec(\xvec) + \Omat^{-1/2}\Gmat^T \epvec $ and $\Var(\Omat^{-1/2}\Gmat^T \epvec) = \bm I$.
	In this way, it only penalizes the regression function for the material part, using the envelope structure to inform the regularization.

	We estimate the response envelope and enveloped regression function by maximizing the penalized log-likelihood, solving
	\begin{equation}
		\argmax_{\Eset, \gvec, \Omat, \Omat_0} \left\{ \ell_u(\Eset, \gvec, \Omat, \Omat_0)  - (\lambda/2) \rho(\gvec, \Omat)\right\} \label{eqn:objective in terms of g}
	\end{equation}
	where $\lambda > 0$ is a tuning parameter.
	While we have assumed $\epvec$ is normally distributed to motivate the derivation of this objective, we do not rely on normality for our theoretical results and thus treat \eqref{eqn:objective in terms of g} as a general-purpose objective for the remainder of this study.
	Through some elementary calculations, we can break the likelihood-based term into components for the material and immaterial parts of $\yvec$, as
	\begin{align*}
		\ell_u(\Eset, \gvec, \Omat, \Omat_0) & = -(nr/2) \log (2 \pi) - (n/2)\log| \Omat_0 | - (1/2) \sum_{i=1}^n \yvec_i^T\Gmat_0 \Omat_0^{-1} \Gmat_0^T \yvec_i \\
		& \hspace{0.5cm} - (n/2)\log| \Omat | - (1/2) \sum_{i=1}^n (\Gmat^T\yvec_i - \gvec(\xvec_i))^T \Omat^{-1} (\Gmat^T \yvec_i - \gvec(\xvec_i))\\
		& = -(nr/2) \log (2 \pi) + L_1( \Eset, \Omat_0) + L_2(\Eset, \gvec, \Omat).
	\end{align*}
	
	Suppose that we hold $\Gmat$ fixed. Partially maximizing $L_2(\Eset, \gvec, \Omat) - (\lambda/2) \rho(\gvec, \Omat)$ with respect to $\gvec$ is an optimization problem over an infinite-dimensional space.
	The following \textit{representer theorem}, however, provides that the solution to \eqref{eqn:objective in terms of g} has a finite-dimensional representation in terms of the kernel functions $\{K(\cdot, \xvec_i)\}_{i \in [n]}$:
	
	\begin{theorem}\label{thm:rep theorem}
		The solution to 
		$\argmax_{\gvec} \{ L_2(\Eset, \gvec, \Omat) - (\lambda/2) \rho(\gvec, \Omat) \}$
		can be represented as
		\begin{equation}
			g_j(\xvec) = \sum_{i=1}^n \alpha_{ij} K(\xvec, \xvec_i) \text{\;\; for \;\;} j \in [u] \text{,\;\;} \alpha_{ij} \in \reals. \label{eqn:rep theorem}
		\end{equation}
	\end{theorem}

	Let $K(\xvec) = [K(\xvec, \xvec_1), \ldots, K(\xvec, \xvec_n)] \in \reals ^{1 \times n}$ and $\amat = [\alpha_{ij}] \in \reals^{u \times n}$.
	By the representer theorem, the partial maximizer can be written as
	$ \gvec(\xvec) =  \amat K(\xvec)^T$.
	Let 
	$\bm K = 
	\left[\begin{smallmatrix} K(\xvec_1, \xvec_1) & \cdots & K(\xvec_1, \xvec_n)\\
		\vdots & \ddots & \vdots\\
		K(\xvec_n, \xvec_1) & \cdots & K(\xvec_n, \xvec_n)
	\end{smallmatrix}\right] 
	= \left[\begin{smallmatrix} K(\xvec_1)\\
		\vdots\\
		 K(\xvec_n)
	\end{smallmatrix}\right]$, the Gram matrix.
	Using this representation, we see that $\norm{[\Omat^{-1/2}]_{i\cdot} \gvec}_{\RKHS}^2 = [\Omat^{-1/2}]_{i\cdot}\amat \bm K \amat^T[\Omat^{-1/2}]_{\cdot i}$ for each $i \in [u]$ . As such, we can write the penalty in terms of $\amat$ as $\rho(\amat, \Omat) = \tr\{\Omat^{-1/2}\amat \bm K \amat^T\Omat^{-1/2}\} = \tr\{\amat \bm K \amat^T\Omat^{-1}\}$.
	We can similarly express the likelihood-based term for the material part as a function of $\amat$, writing
	$
		L_2(\Eset, \amat, \Omat) = - (n/2)\log| \Omat | - (1/2) \sum_{i=1}^n (\Gmat^T\yvec_i - \amat K(\xvec_i)^T)^T \Omat^{-1} (\Gmat^T \yvec_i - \amat K(\xvec_i)^T)
	$.
	
	It is straightforward to show that for a fixed $\Gmat$ the solution to $\argmax_{\amat} L_2(\Eset, \amat, \Omat) - (\lambda/2 )  \tr\{\amat \bm K \amat^T\Omat^{-1}\}$
	has a closed form: $\amat = \Gmat^T \Ymat^T (\bm K + \lambda \bm I)^{-1}$.
	By extension, the maximizing regression function is $\gvec(\xvec) = \Gmat^T \Ymat^T (\bm K + \lambda \bm I)^{-1} K (\xvec)^T$.
	Plugging-in this solution and partially maximizing with respect to $\Omat$ yields $\Omat = \frac{1}{n} \Gmat^T 
	(\Ymat^T\Ymat - \Ymat^T \bm K(\bm K + \lambda \bm I)^{-1} \Ymat )\Gmat = \Gmat^T 
	\bm S_{\Ymat|\bm K}^{\lambda}\Gmat$, where $\bm S_{\Ymat|\bm K}^{\lambda} =  \frac{1}{n} \Ymat^T\Ymat -  \frac{1}{n}\Ymat^T \bm K(\bm K + \lambda \bm I)^{-1} \Ymat$.
	Plugging these expressions for the maximizing $\gvec$ and $\Omat$ into $L_2(\Eset, \gvec, \Omat) - (\lambda/2) \rho(\gvec, \Omat)$, we get the following partially-maximized expression for the material component of the objective function:
	$$
		\mathcal L_2(\Eset) = -(n/2) \log | \Gmat^T 
		\bm S_{\Ymat|\bm K}^{\lambda}\Gmat |,
	$$
	which only depends on $\Gmat$.
	
	We follow similar steps to simplify $\mathcal{L}_1(\Eset, \Omat_0)$.
	Partially-maximizing with respect to $\Omat_0$ yields
	$\Omat_0 = \Gmat_0^T \bm S_{\Ymat} \Gmat_0$ where $\bm S_{\Ymat} = \frac{1}{n} \Ymat^T\Ymat$. 
	Plugging this in, we find
	$$
		\mathcal L_1(\Eset) = -(n/2) \log | \Gmat_0^T \bm S_{\Ymat}\Gmat_0 | =  - (n/2) \log | \bm S_{\Ymat} | -(n/2) \log | \Gmat^T \bm S_{\Ymat}^{-1} \Gmat |.
	$$
	
	Combining these partially-maximized terms, we see that estimating $\envlp$ amounts to solving
	\begin{equation}
		\hatenvlp = \spann \{ \argmin_{\bm G \in \text{Gr}(r,u)} \log| \bm G^T \bm S_{\Ymat}^{-1} \bm G|  + \log|  \bm G^T \bm S_{\Ymat|\bm K}^{\lambda} \bm G|  \}, \label{eqn:grassmanian optim}
	\end{equation}
	where $\text{Gr}(r,u) = \{\bm G \in \reals^{r \times u} : \bm G \text{ is a semi-orthogonal matrix}\}$.
	The optimal semi-orthogonal matrix $\bm G$ for \eqref{eqn:grassmanian optim} may not be unique, but the estimated envelope subspace $\hatenvlp$ will be.
	\cite{Cook2016} proposed a fast iterative algorithm for solving Grassmanian optimization problems of the form of \eqref{eqn:grassmanian optim}.
	We follow their approach to compute $\hatenvlp$.

	Because we derived \eqref{eqn:grassmanian optim} by partially maximizing the objective with respect to $\gvec$, $\Omat_0$, and $\Omat_0$ for a fixed $\Gmat$, we obtain our estimates of those parameters using an estimated basis for $\hatenvlp$.
	Let $\hatGmat$ denote any semi-orthogonal basis matrix for $\hatenvlp$ and let $\hatGmat_0$ denote any semi-orthogonal basis matrix for the orthogonal complement of $\hatenvlp$.
	Given a basis $\hatGmat$, our estimate of $\amat$ is $\hat{\amat}^{\lambda} = (\hatGmat)^T \Ymat^T(\bm K + \lambda \bm I)^{-1}$, our kernelized estimate of $\gvec$ is $\hat \gvec^{\lambda}(\xvec) = \hat \amat^{\lambda} K(x)^T$, and our estimate of $\fvec$ is $\hat \fvec^{\lambda}(\xvec) = \hatGmat \hat \gvec^{\lambda}(\xvec)$.
	Our estimates for the covariance matrices from the response envelope model \eqref{eqn:envelope model} are
	$\widehat \Omat^{\lambda} = (\hatGmat)^T \bm S_{\Ymat|\bm K}^{\lambda} \hatGmat$, 
	$\widehat{\Omat}_0 = (\hatGmat_0)^T \bm S_{\Ymat} \hatGmat_0$, and 
	$\widehat \Sigmat = \hatGmat \widehat{\Omat}^{\lambda} (\hatGmat)^T + \hatGmat_0 \widehat{\Omat} (\hatGmat_0)^T$.
	We refer to this nonparametric estimator of the response envelope and regression function in an RKHS as the \textit{kernel envelope} (KENV) estimator.
	
	We can express the KENV estimate of $\fvec(\xvec)$ succinctly as $\hat \fvec^{\lambda}(\xvec) = \hatGmat (\hatGmat)^T \Ymat^T(\bm K + \lambda \bm I)^{-1}  K(\xvec)^T$.
	We recognize this as the projection of the kernel ridge regression estimator, $\tilde \fvec(\xvec) = \Ymat^T(\bm K + \lambda \bm I)^{-1}  K(\xvec)^T$, onto the estimated envelope subspace $\hatenvlp$. 
	In this relation, we see that the KENV estimated regression function $\hat \fvec^{\lambda}(\xvec)$ depends only on the estimated subspace $\hatenvlp$, not on the specific choice of basis.

	\subsection{Relation to existing methods}
	KENV generalizes several existing multivariate regression methods.
	\citet{kwon_enhanced_2024}'s enhanced envelope estimator can be viewed as KENV with a linear kernel.
	Likewise, kernel ridge regression is a special case of KENV with $u = r$. 

	One key difference between KENV and the martingale difference divergence envelope (MDDE) and central mean envelope (CME) estimators of \citet{zhang_envelopes_2020} is that KENV directly and simultaneously estimates the envelope $\envlp$ and the corresponding enveloped regression function $\fvec$ from the nonparametric envelope model \eqref{eqn:envelope model} while MDDE and CME do not provide an estimate of $\fvec$.
	The MDDE and CME estimators are model-free and defined in terms of the expectation and covariance---in particular, they estimate subspaces $\mathcal S$ and $\mathcal T$ such that the projections of $\yvec$ onto the subspaces and their orthogonal complements satisfy $\E(\bm Q_{\mathcal S}\yvec| \xvec) = \E(\bm Q_{\mathcal S}\yvec )$ and $\Cov(\bm P_{\mathcal S} \yvec, \bm Q_{\mathcal S} \yvec) = 0$, in the case of the MDDE, and $\E(\bm Q_{\mathcal T}\yvec| \xvec) = \E(\bm Q_{\mathcal T}\yvec )$ and $\Cov(\bm P_{\mathcal T} \yvec, \bm Q_{\mathcal T} \yvec|\xvec) = 0$, in the case of the CME.
	As such, their formulations do not include a regression function that needs to be estimated.
	For KENV, on the other hand, the envelope and regression function are linked: the immaterial part of $\yvec$ is exactly the portion of $\yvec$ that cannot be explained by $\fvec(\xvec)$ per conditions \ref{condition i} and \ref{condition ii}.

	\subsection{Implementation details}%

	In practice, the envelope dimension $u$ is typically unknown.
	We treat it as a hyperparameter for KENV which must be tuned, along with $\lambda$ and, in some cases, a kernel hyperparameter $\sigma$.
	We perform a grid search over ranges of hyperparameter values and tune them with $M$-fold cross-validation (CV).
	
	For fixed $u$ and $\sigma$, we accelerate computation by fitting KENV for a sequence of penalty parameter values $\lambda_1 > \ldots > \lambda_L$.
	Evaluating the kernel function $K(\cdot, \cdot)$ can be computationally costly.
	Because $\lambda$ has no bearing on $K(\cdot, \cdot)$, however, we can compute the gram matrix $\bm K$ once for the entire sequence. %
	Likewise, when making a prediction at a new value $\xvec_{new}$, we can compute $K(\xvec_{new})$ once, then obtain $\hat \fvec^{\lambda_l} (\xvec_{new}) = \widehat{ \bm \Gmat}^{\lambda_l} \hat{\amat}^{\lambda_l} K(\xvec_{new})^T $ for each $\lambda_l$ in the sequence.
	We further speed-up model fitting by warm-starting computation of the envelope subspace---that is, for $\lambda_l$, we initialize the iterative algorithm to solve \eqref{eqn:grassmanian optim} with the estimated basis for the previous $\lambda$ value, $\hat \Gmat^{\lambda_{l-1}}$.

	\section{Theoretical Study}\label{sec:theory}
	In this section, we study the prediction risk of the KENV estimator.
	First we compare the in-sample prediction risk of KENV with that of kernel ridge regression (KRR). 
	Then we study the limiting risk of the KENV estimator as $n \to \infty$.
	
	\subsection{Comparison of in-sample risk with kernel ridge regression}
	For an estimator $\hat \fvec$ of $\fvec$, we define the in-sample prediction risk to be:
	\begin{equation}
		R(\hat \fvec (\Xmat)|\Xmat) = \E\left[ \norm{ \hat \fvec(\Xmat) - \fvec(\Xmat) }_F^2 \middle| \Xmat \right], \label{eqn:in sample risk} 
	\end{equation}
	where $\Xmat$ denotes the design matrix from the training data, $\fvec(\Xmat) = [\fvec(\xvec_1) \cdots \fvec(\xvec)]^T$, $\hat \fvec(\Xmat)$ is defined equivalently, and the expectation is taken with respect to $\Ymat$.
	We can decompose the in-sample risk into bias and variance components:
	\begin{align}
		\E\left[ \norm{ \hat \fvec(\Xmat) - \fvec(\Xmat) }_F^2 \middle| \Xmat \right]
		& = \E\left[ \norm{ \vect( \hat \fvec(\Xmat)) - \vect( \fvec(\Xmat) ) }_2^2 \middle| \Xmat \right] \notag\\
		& = \norm{ \vect( \E[  \hat \fvec(\Xmat)| \Xmat ] - \hat \fvec(\Xmat) ) }_2^2 + \tr\{ \Var[ \vect( \hat \fvec(\Xmat)) |\Xmat ] \} \label{eqn:bv decomp}.
	\end{align}
	
	To facilitate our theoretical analysis, we assume that a basis $\Gmat$ for $\envlp$ is known. 
	This has been done in prior studies of envelope methods to understand their underlying mechanisms \citep{Cook_foundations_2015, Cook_simultaneous_2015, kwon_enhanced_2024}.
	As a consequence of this assumption, these results do not account for the variability arising from estimating of $\Gmat$. 
	Nevertheless, they provide some insight into the relative performances of KENV and KRR.
	
	Let $\hat \fvec^R(\Xmat)$ denote the KRR estimator and 	$\hat \fvec^E_{\Gmat}(\Xmat)$ denote the KENV estimator with envelope basis $\Gmat$.
	The following result compares their in-sample prediction risks:
	\begin{prop}\label{prop:in sample risk comparison}
		Under a nonparametric response envelope model \eqref{eqn:envelope model},
		$ R( \hat \fvec^E_{\Gmat}(\Xmat) | \Xmat ) \leq R( \hat \fvec^R(\Xmat) | \Xmat )$, with a strict inequality if $u < r$ and $\tr(\Omat_0) > 0$.
	\end{prop}
	Proposition \ref{prop:in sample risk comparison} establishes that, when $\Gmat$ is known, the in-sample risk of the KENV estimator never exceeds that of the KRR estimator. 
	If the response is generated under a nonparametric response envelope model \eqref{eqn:envelope model} with a non-constant immaterial component (i.e. $\Omat_0 \neq \bm 0$), then the risk for the KENV estimator is guaranteed to be strictly smaller than that of KRR.
	
	In the proof of the theorem, we find that KENV attains a lower risk by reducing the variance while achieving the same bias as KRR.
	Naturally, we expect that some bias and additional variance will arise from the estimation of $\Gmat$ in practice.
	However, Proposition \ref{prop:in sample risk comparison} provides some intuition for how KENV can improve the estimation of $\bm f$ given the envelope basis.
	Moreover, we find in our empirical study (Sections \ref{sec:sims} and \ref{sec:real data}) that KENV consistently delivers more accurate predictions than KRR in practice, suggesting that the cost of estimating $\Gmat$ is offset by the resulting variance reduction in estimating $\bm f$.
	
	\subsection{Limiting risk of KENV}
	For the remainder of the theoretical study, we investigate the risk of KENV as the sample size diverges to infinity. 
	For a loss function $L(\yvec, \fvec(\xvec))$ we define the \textit{risk} to be
	$
		R(\bm f) = \E[ L(\bm y,  \fvec(\xvec)) ] 
	$
	and the \textit{optimal risk} to be
	$
		R^* = \inf \{ R(\bm f) \mid \fvec:\reals^p \to \reals^r \text{ is measurable}\}.
	$
	We say that an estimator $\hat \fvec$ is \textit{risk consistent} if its risk converges to the optimal risk in probability as $n \to \infty$.
	
 	To facilitate our study, we assume $\Gmat$ and $\Omat$ are known.
 	We allow the penalty parameter $\lambda_n$ in \eqref{eqn:objective in terms of g} to vary with $n$.
 	To establish the risk consistency of the KENV solution, we assume the kernel function is \textit{universal}:
 	\begin{definition}
 		Suppose that $\mathcal X \in \reals^p$ is compact.
 		We say that a continuous kernel $K(\cdot, \cdot): \mathcal X \times \mathcal X \mapsto \reals$ is \textit{universal} if the RKHS generated by $K$ is dense in the space of continuous functions on $\mathcal X$ with respect to the supremum norm $\norm{f}_{\infty} = \sup_{x \in \mathcal X} |f(x)|$.
 	\end{definition}
 	Intuitively, if $K$ is a universal kernel, then for any continuous function $f:\mathcal{X} \to \reals$ there exists a function $h \in \RKHS$ that is arbitrarily close to $f$.
 	Many popular kernels are known to be universal, including the Gaussian and exponential kernels \citep{steinwart_influence_2001}. 
 	\citet{steinwart_influence_2001} and \cite{christmann_consistency_2007} established the consistency of support vector machine and kernel regression estimators with universal kernels. 	
 	We similarly leverage the properties of universal kernels to show that the KENV solution is consistent under the mild assumption that the responses have finite second moments.
 	\begin{theorem}\label{thm:risk consistency}
 		Suppose that $\mathcal X \in \reals^p$ is compact, that $\mathcal H_K$ is an RKHS of a universal kernel on $\mathcal X$, and that $\E[ y_j^2 ]  < \infty$ for $j = 1, \ldots, r$.
 		If the penalty parameter $\lambda_n$ satisfies $\lambda_n/n \to 0$ and $\lambda_n/n^{3/4} \to \infty$ as $n \to \infty$, then the KENV solution satisfies
 		$
 			R(\hat \gvec) \to_p R^* \text { as } n \to \infty.
 		$ 
 	\end{theorem}
 	This result establishes the risk consistency of the KENV estimator in the case where $\Gmat$ and $\Omat$ are known.
 	This consistency result separates KENV from existing response envelope methods which require the true model to be linear for their convergence guarantees to hold \citep{Cook2010, kwon_enhanced_2024}.
 	Unlike those parametric envelope estimators, KENV achieves risk consistency under the general envelope model \eqref{eqn:envelope model} which encompasses both the parametric and nonparametric cases.
 	For modelers unsure as to whether the true multivariate model is linear, KENV is a flexible method with asymptotic risk guarantees that hold under far more general model assumptions than alternative envelope methods.
 		
	\section{Simulations} \label{sec:sims}

	We conduct several simulation studies to evaluate the finite-sample performance of KENV.
	First we compare the prediction performance of KENV with that of other multivariate regression methods under different data-generating models.
	Following that we compare the fitted curves from KENV and kernel ridge regression to gain a better understanding of how KENV reduces the prediction risk.
	
	\subsection{Prediction comparison}
	For our first set of simulations, we compare the prediction performance of KENV with that of \citet{kwon_enhanced_2024}'s enhanced envelope estimator (EENV), kernel ridge regression (KRR), ridge regression (RR), and a support vector machine (SVM) \citep{cortes_support-vector_1995}.
	We include KENV and KRR models with the Gaussian and Laplacian kernels---for which $K(\xvec, \xvec') = \exp(- \norm{\xvec - \xvec'}_2^2/\sigma^2)$ and $K(\xvec, \xvec') = \exp(- \norm{\xvec - \xvec'}_2/\sigma)$, respectively---in the comparison, denoting the Gaussian kernel models by KENV(G) and KRR(G) and the Laplacian kernel models by KENV(L) and KRR(L).
	We use the Gaussian kernel for the SVM, fitting it via $r$ univariate regressions with the \href{https://cran.r-project.org/web/packages/e1071/index.html}{e1071} R package.
	We tune $u$, $\lambda$, and $\sigma$ using $5$-fold CV on the training data.
	
	\subsubsection{Response envelope data generation}
	We examine several data-generating models in our simulations.
	In every case, the true model is a nonparametric response envelope \eqref{eqn:envelope model}.
	The procedures for generating $\xvec$ and $\gvec$ vary across simulation settings and will be described later.
	Given $ \{\xvec_i\}_{i \in [n]} $ and $\gvec$, we generate $\{\yvec_i\}_{i \in [n]}$ with an envelope structure as follows:\\
	\indent \textit{Input:} $\mathcal{P} \subseteq [r]$, $\Omat$, $\Omat_0$
	\begin{enumerate}[leftmargin=0.5in]
		
		\item Generate a random orthogonal matrix $\bm V \in \reals^{r \times r}$ of eigenvectors of $\Sigmat$.
		\item Set $\Gmat = \bm V_{(\mathcal P)}$ and $\Gmat_0 = \bm V_{([r]\setminus\mathcal{P} )}$. Then $\Sigmat =  \Gmat \Omat \Gmat^T + \Gmat_0 \Omat_0 \Gmat_0^T$.
		\item Generate $\yvec_i =  \Gmat \gvec(\xvec_i) +  \epvec_i$ where $\epvec_i \sim N(\bm 0, \Gmat \Omat \Gmat^T + \Gmat_0 \Omat_0 \Gmat_0^T)$ for $i \in [n]$.
	\end{enumerate}
	
	\indent \textit{Return:} $\{\yvec_i\}_{i \in [n]}$, $\fvec = \Gmat \gvec$\\
	Under this model $\mathcal{P}$ denotes the indices of the columns of $\bm V$ that comprise the response envelope $\envlp$.
	We vary $r,u,p, \Omat$, and $\Omat_0$ across simulation settings to create different envelope structures.

	\subsubsection{Single predictor simulations}
	For our first comparison, we generate data with a single predictor.
	We examine two data-generating models:
	\begin{itemize}[leftmargin=0.25in]
		
		\item \textbf{Model 1:} $r = 3, u = 2, \Omat = \left[\begin{smallmatrix}
			4 & 0\\
			0 & 2
		\end{smallmatrix}\right], \Omat_0 = \left[	5 \right]$,
		$g_1(x) = 2 \sin(x) + x^2/5 - x/2$, and 
		$g_2(x) = \cos(x) - x^2/10 + x/3 $.
		\item \textbf{Model 2:} $r = 4, u = 1, \Omat = \left[	4 \right], \Omat_0 = \left[\begin{smallmatrix}
			5 & 0 & 0\\
			0 & 2 & 0 \\
			0 & 0 & 1
		\end{smallmatrix}\right]$, and
		$
			g_1(x) = 2 \sin(x) + x^2/5 - x/2.
		$
	\end{itemize}
	In both cases, we generate $x_i \iid \text{Uniform}(-5,5)$ for $i \in [n]$.
	Within each of these cases, we run simulations with $n = 100, 200,$ and $400$.
	
	To assess the prediction performance, we generate a test set of $n_{test} = 2000$ observations and compute the mean squared error $MSE(\hat \fvec) = \frac{1}{n_{test}}\sum_{i' = 1}^{n_{test}} \sum_{j = 1}^r(\hat f_j(x_{i'}) - f_j(x_{i'}))^2$ and mean absolute error $MAE(\hat \fvec) =  \frac{1}{n_{test}}\sum_{i' = 1}^{n_{test}} \sum_{j = 1}^r | \hat f_j(x_{i'}) - f_j(x_{i'}) |$,
	for each estimator $\hat \fvec = (\hat f_1, \ldots, \hat f_r)^T$.
	We run $100$ replications in each setting and report the mean and standard error of MSE and MAE across replications.
	We further report the proportion of replications in which the selected envelope dimensions for KENV and EENV coincide with the true dimension $u$, denoting this by $\#u$ in the results.
	
	Tables \ref{tbl:univariate case 1} and \ref{tbl:univariate case 2} report the prediction results from the single predictor simulations with Models 1 and 2, respectively.
	The best performing method is highlighted in bold in each row.
	Some general patterns hold for both data-generating models and for all values of $n$:
	KENV delivers more accurate predictions than KRR with both the Gaussian and Laplace kernels in every case.
	Among the three kernels examined, the Gaussian kernel consistently delivers the most accurate predictions, followed closely by the Laplacian kernel.
	Moreover, KENV(G) and KENV(L) outperform SVM in every case.
	
	\spacingset{1.15}
	\begin{table}
		\caption{Prediction comparison - single predictor model 1}
		\label{tbl:univariate case 1}
		\centering
		\resizebox{\columnwidth}{!}{
			\begin{tabular}{llllllll}
				\toprule
				& KENV(G) & KRR(G) & KENV(L) & KRR(L) & EENV & RR & SVM\\
				\midrule
				\addlinespace[0.3em]
				\multicolumn{8}{l}{{n = 100}}\\
				\hline\hspace{1em}MSE & \textbf{0.62 (0.04)} & 0.71 (0.04) & 0.71 (0.03) & 0.85 (0.03) & 5.98 (0.03) & 5.97 (0.03) & 0.92 (0.04)\\
				\hspace{1em}MAE & \textbf{1.04 (0.03)} & 1.12 (0.03) & 1.12 (0.02) & 1.24 (0.02) & 3.3 (0.02) & 3.3 (0.02) & 1.3 (0.03)\\
				\hspace{1em}\#u & 0.81 & - & \textbf{0.85} & - & 0.3 & - & -\\
				\addlinespace[0.3em]
				\multicolumn{8}{l}{{n = 200}}\\
				\hline\hspace{1em}MSE & \textbf{0.31 (0.01)} & 0.37 (0.01) & 0.41 (0.01) & 0.51 (0.01) & 5.83 (0.02) & 5.82 (0.02) & 0.5 (0.02)\\
				\hspace{1em}MAE & \textbf{0.73 (0.02)} & 0.82 (0.02) & 0.85 (0.01) & 0.96 (0.01) & 3.29 (0.02) & 3.29 (0.02) & 0.94 (0.02)\\
				\hspace{1em}\#u & \textbf{0.84} & - & 0.83 & - & 0.46 & - & -\\
				\addlinespace[0.3em]
				\multicolumn{8}{l}{{n = 400}}\\
				\hline\hspace{1em}MSE & \textbf{0.15 (0.01)} & 0.19 (0.01) & 0.22 (0.01) & 0.27 (0.01) & 5.75 (0.01) & 5.75 (0.01) & 0.24 (0.01)\\
				\hspace{1em}MAE & \textbf{0.51 (0.01)} & 0.58 (0.01) & 0.62 (0.01) & 0.71 (0.01) & 3.26 (0.02) & 3.26 (0.02) & 0.65 (0.01)\\
				\hspace{1em}\#u & \textbf{0.88} & - & 0.87 & - & 0.59 & - & -\\
				\bottomrule
			\end{tabular}
		}
	\end{table}
	
	\begin{table}
		\caption{Prediction comparison - single predictor model 2}
		\label{tbl:univariate case 2}
		\centering
		\resizebox{\columnwidth}{!}{
			\begin{tabular}{llllllll}
				\toprule
				& KENV(G) & KRR(G) & KENV(L) & KRR(L) & EENV & RR & SVM\\
				\midrule
				\addlinespace[0.3em]
				\multicolumn{8}{l}{{n = 100}}\\
				\hline\hspace{1em}MSE & \textbf{0.62 (0.04)} & 0.78 (0.03) & 0.68 (0.03) & 0.92 (0.03) & 4.56 (0.02) & 4.57 (0.02) & 0.99 (0.04)\\
				\hspace{1em}MAE & \textbf{1.16 (0.03)} & 1.35 (0.03) & 1.22 (0.03) & 1.47 (0.02) & 3 (0.03) & 3.01 (0.03) & 1.5 (0.03)\\
				\hspace{1em}\#u & 0.72 & - & \textbf{0.76} & - & 0.65 & - & -\\
				\addlinespace[0.3em]
				\multicolumn{8}{l}{{n = 200}}\\
				\hline\hspace{1em}MSE & \textbf{0.27 (0.02)} & 0.36 (0.02) & 0.33 (0.02) & 0.48 (0.01) & 4.36 (0.02) & 4.37 (0.02) & 0.5 (0.02)\\
				\hspace{1em}MAE & \textbf{0.75 (0.02)} & 0.92 (0.02) & 0.83 (0.02) & 1.08 (0.02) & 2.92 (0.02) & 2.93 (0.02) & 1.08 (0.02)\\
				\hspace{1em}\#u & 0.81 & - & \textbf{0.89} & - & 0.69 & - & -\\
				\addlinespace[0.3em]
				\multicolumn{8}{l}{{n = 400}}\\
				\hline\hspace{1em}MSE & \textbf{0.15 (0.01)} & 0.2 (0.01) & 0.18 (0.01) & 0.29 (0.01) & 4.28 (0.01) & 4.28 (0.01) & 0.27 (0.01)\\
				\hspace{1em}MAE & \textbf{0.56 (0.01)} & 0.68 (0.01) & 0.62 (0.01) & 0.84 (0.01) & 2.84 (0.03) & 2.84 (0.03) & 0.8 (0.01)\\
				\hspace{1em}\#u & 0.83 & - & \textbf{0.9} & - & 0.61 & - & -\\
				\bottomrule
			\end{tabular}
		}
	\end{table}
	\spacingset{1.7}
	
	EENV and RR, which model a linear relationship between the responses and predictors, perform much worse than the kernel methods.
	Notably, EENV does not improve on RR despite the fact that the response has an envelope structure, suggesting that its performance is degrading due to misspecification of the regression function.
	
	Comparing results across the data-generating models, we see that the gap between the KENV and KRR models is larger for Model 2 than for Model 1, providing MSE reductions on the order of 25\% for both the Gaussian and Laplacian kernels.
	This likely reflects the fact that the immaterial part of $\yvec$ accounts for a greater proportion of the variability under Model 2 than under Model 1 so there is more to be gained by removing the immaterial part.

	\subsubsection{Multiple predictor simulations}
	For our second prediction comparison, we generate data with multiple predictors: $\xvec \in \reals^p$.
	We generate the components of the true reduced regression function $\gvec$ under the ``random function generator'' model \citep{friedman_greedy_2001}, as follows:\\
	\indent \textit{Input:} $p, u$
	\begin{enumerate}[leftmargin=0.5in]
		
		\item For $j \in [u]$:
		\begin{enumerate}
			
			\item for $l \in [20]$:
			\begin{enumerate}
				\item Randomly sample $p_{jl}$ predictors $\xvec_{jl}$ from $\xvec$, where $p_{jl} = \min( \lfloor 1.5 + T_{jl} \rfloor, p  )$ with $T_{jl} \sim \text{Exp}(0.5)$.
				\item Generate a random coefficient $a_{jl} \sim \text{Unif}[-10, 10]$.
				\item Generate a random function $g_{jl}(\xvec_{jl})$ as
				$$
					g_{jl}(\xvec_{jl}) = \exp \left[ -\frac{1}{2} (\xvec_{jl} - \bm \mu_{jl} )^T \bm V_{jl} (\xvec_{jl} - \bm \mu_{jl} ) \right]
				$$
				where
				$ \bm \mu_{jl} \sim N(\bm 0, \bm I_{p_{jl}})$ and $\bm V_{jl}  = \bm U_{jl} \bm D_{jl} \bm U_{jl}^T$ with
				$\bm U_{jl}$ being a randomly generated $p_{jl} \times p_{jl}$ orthogonal matrix and $D_{jl} = \diag( d_{1l}, \ldots, d_{p_{jl} l})$ where $\sqrt{d_{kl}} \sim \text{Unif}[0.1, 2]$ for $k \in [p_{jl}]$.
				
			\end{enumerate}
			\item Set $g_j(\xvec) = \sum_{l = 1}^{20} a_{jl} g_{jl} (\xvec_{jl})$.
		\end{enumerate}
	\end{enumerate}
	
	\indent \textit{Return:} $\bm \gvec(\xvec)  = (g_1(\xvec), \ldots, g_u (\xvec))^T $.
	
	For the response envelope structure, we set $r = 4, u = 1, \Omat = \left[	4 \right],$ and $\Omat_0 = \left[\begin{smallmatrix}
		5 & 0 & 0\\
		0 & 2 & 0 \\
		0 & 0 & 1
	\end{smallmatrix}\right]$.
	We generate $\xvec_i \iid N( \bm 0, \Sigmat_{\xvec})$ with $[\Sigmat_{\xvec}]_{ij}  = \rho^{|i-j|}$ for $i,j \in [p]$, examining cases where $\rho = 0$ and $\rho = 0.8$.
	We set $n = 200$ in every case and run simulations with $p = 10, 200$, and $400$, giving us both low and high-dimensional settings.
	As in the single predictor simulations, we generate an independent test set of $n_{test} = 2000$ observations and collect the MSE, MAE, and $\#u$ from $100$ replications to assess model performance.

	Tables \ref{tbl:multivariate rho 0} and \ref{tbl:multivariate rho 0.8} report the results from the multiple predictor simulations with $\rho = 0$ and $\rho = 0.8$, respectively.	
	As in the single predictor case, KENV dominates KRR with both the Gaussian and Laplacian kernels.
	This pattern holds across both the low and high-dimensional settings, though the differences in prediction performance are more pronounced in low dimensions.
	SVM is more competitive with KRR than in the single predictor simulations, attaining lower test MSEs than KRR(G) and KRR(L) in several cases.
	KENV, however, still outperforms SVM in every case.
	As before, the kernel methods clearly outperform the linear methods.
	
	\spacingset{1.15}	
	\begin{table}
		\caption{Prediction comparison - multiple predictors with $\rho = 0$}
		\label{tbl:multivariate rho 0}
		\centering
		\resizebox{\columnwidth}{!}{
		\begin{tabular}{llllllll}
			\toprule
			& KENV(G) & KRR(G) & KENV(L) & KRR(L) & EENV & RR & SVM\\
			\midrule
			\addlinespace[0.3em]
			\multicolumn{8}{l}{{p = 10}}\\
			\hline\hspace{1em}MSE & \textbf{14.25 (0.48)} & 16.05 (0.51) & 15.43 (0.57) & 16.38 (0.57) & 25.39 (1.1) & 25.64 (1.1) & 16.05 (0.52)\\
			\hspace{1em}MAE & \textbf{5 (0.1)} & 5.64 (0.1) & 5.21 (0.11) & 5.61 (0.11) & 6.71 (0.17) & 6.82 (0.16) & 5.51 (0.11)\\
			\hspace{1em}\#u & \textbf{1} & - & 0.97 & - & 0.85 & - & -\\
			\addlinespace[0.3em]
			\multicolumn{8}{l}{{p = 200}}\\
			\hline\hspace{1em}MSE & \textbf{32.41 (1)} & 32.78 (1) & 32.59 (1) & 32.83 (1) & 35.92 (1.17) & 38.67 (1.19) & 32.75 (1.01)\\
			\hspace{1em}MAE & \textbf{7.66 (0.14)} & 7.79 (0.13) & 7.68 (0.14) & 7.77 (0.13) & 8.07 (0.15) & 8.77 (0.13) & 7.72 (0.14)\\
			\hspace{1em}\#u & 0.62 & - & 0.6 & - & \textbf{0.87} & - & -\\
			\addlinespace[0.3em]
			\multicolumn{8}{l}{{p = 400}}\\
			\hline\hspace{1em}MSE & \textbf{33.3 (1.06)} & 33.45 (1.07) & 33.39 (1.07) & 33.49 (1.07) & 37.45 (1.18) & 41.11 (1.2) & 33.52 (1.07)\\
			\hspace{1em}MAE & \textbf{7.8 (0.15)} & 7.85 (0.15) & 7.81 (0.15) & 7.84 (0.15) & 8.33 (0.16) & 9.08 (0.14) & 7.84 (0.15)\\
			\hspace{1em}\#u & 0.4 & - & 0.36 & - & \textbf{0.84} & - & -\\
			\bottomrule
		\end{tabular}
		}
	\end{table}
	
	\begin{table}
		\caption{Prediction comparison - multiple predictors with $\rho = 0.8$}
		\label{tbl:multivariate rho 0.8}
		\centering
		\resizebox{\columnwidth}{!}{
		\begin{tabular}{llllllll}
			\toprule
			& KENV(G) & KRR(G) & KENV(L) & KRR(L) & EENV & RR & SVM\\
			\midrule
			\addlinespace[0.3em]
			\multicolumn{8}{l}{{p = 10}}\\
			\hline\hspace{1em}MSE & \textbf{7.43 (0.25)} & 9.16 (0.28) & 7.85 (0.28) & 9.11 (0.28) & 25.34 (1.18) & 25.52 (1.17) & 9.16 (0.28)\\
			\hspace{1em}MAE & \textbf{3.6 (0.07)} & 4.33 (0.06) & 3.7 (0.07) & 4.3 (0.06) & 6.72 (0.16) & 6.8 (0.16) & 4.22 (0.07)\\
			\hspace{1em}\#u & \textbf{1} & - & \textbf{1} & - & 0.75 & - & -\\
			\addlinespace[0.3em]
			\multicolumn{8}{l}{{p = 200}}\\
			\hline\hspace{1em}MSE & \textbf{29.5 (1)} & 30.18 (1.01) & 29.71 (1.01) & 30.13 (1.01) & 32.48 (1.15) & 34.2 (1.15) & 29.96 (1)\\
			\hspace{1em}MAE & \textbf{7.17 (0.16)} & 7.42 (0.15) & 7.2 (0.16) & 7.37 (0.15) & 7.51 (0.17) & 8.07 (0.16) & 7.29 (0.16)\\
			\hspace{1em}\#u & 0.89 & - & 0.83 & - & \textbf{0.94} & - & -\\
			\addlinespace[0.3em]
			\multicolumn{8}{l}{{p = 400}}\\
			\hline\hspace{1em}MSE & \textbf{31.78} (1.05) & 32.26 (1.04) & 31.98 (1.05) & 32.3 (1.05) & 36.81 (1.23) & 40.13 (1.26) & 32.27 (1.05)\\
			\hspace{1em}MAE & \textbf{7.6 (0.17)} & 7.75 (0.16) & 7.62 (0.17) & 7.74 (0.16) & 8.17 (0.18) & 8.98 (0.16) & 7.69 (0.17)\\
			\hspace{1em}\#u & 0.68 & - & 0.58 & - & \textbf{0.97} & - & -\\
			\bottomrule
		\end{tabular}
		}
	\end{table}
	
	\spacingset{1.7}
	
	\subsubsection{Prediction comparison takeaways}
	The prediction comparison results demonstrate a few key advantages of KENV over competing methods.
	By expanding the flexibility of response envelope models, KENV was able to dominate EENV's prediction performance in every simulation setting.
	EENV, on the other hand, barely improved on ridge regression in the single predictor simulations and even delivered less accurate predictions in some cases, revealing that model misspecification can significantly hamper the performance of linear envelope methods and demonstrating the need for nonparametric envelope methods.
	
	KENV differs from the other nonparametric methods, KRR and SVM, in that it capitalizes on the envelope structure of the multivariate response.
	In doing so, KENV was able to consistently deliver more accurate predictions than KRR and SVM.
	Moreover, these performance gains held with both the Gaussian and Laplacian kernels, suggesting that the improvement over KRR is robust to different choices for the kernel, provided that the kernel is flexible enough to approximate the true relationship.
	
	\subsection{Fitted curve comparison}\label{subsec:fitted curves}
	To better understand how KENV improves on the prediction performance of KRR, we compare their fits with the true regression function on simulated data.
	By examining their fits across many simulation replications, we seek to gain some insight into the comparative bias and variance of the KENV and KRR estimators.
	
	We generate the data using Model 2 from the single predictor comparison.
	We generate one set of $n = 200$ predictor values to use across all simulation replications.
	In each replication, we generate a new realization of the responses and fit KENV and KRR models with the Gaussian kernel to that realization of the data.

	Figure \ref{fig:univariate fits} shows the means ($\pm$ two standard deviations) of the KENV(G) and KRR(G) fits for each component of the response over $100$ simulation replications.
	The true regression functions $f_j(x)$, $j = 1, \ldots, 4$, are shown as dashed lines in each of the plots.
	The curves are plotted over one realization of the simulated data for illustrative purposes.

	Examining the fits for the first three components of $\fvec$, we see that the KENV(G) estimates have lower variance than the KRR(G) estimates.
	Comparing the fits for $ f_4$, we see that the variances of the KENV(G) and KRR(G) estimators are closer than for the other components, and that KRR(G) even has slightly lower variance for some values of $x$.
	However, KENV(G) appears to have lower bias in this case, with the mean of the fitted KENV(G) curves cleaving more closely to the true regression function than the mean of the KRR(G) curves, particularly in places where the true function bends sharply.
	
	Figure \ref{fig:univariate fits} reveals that KENV(G)'s prediction performance gains relative to KRR(G) stem largely from reducing the variance.
	This finding aligns with the theoretical results in Proposition \ref{prop:in sample risk comparison} as well as our broader intuition about response envelopes---namely, that by separating the material and immaterial parts of $\yvec$ we can reduce the variance of our estimator.
	It is significant that we see this variance reduction under a nonparametric model \eqref{eqn:regression model}, as it demonstrates that by applying envelope techniques to nonparametric estimations we can substantially improve their prediction performance.

	\begin{figure}
		\centering
		\caption{Comparing KENV(G) and KRR(G) fits}
		\label{fig:univariate fits}
		\includegraphics[width=\columnwidth]{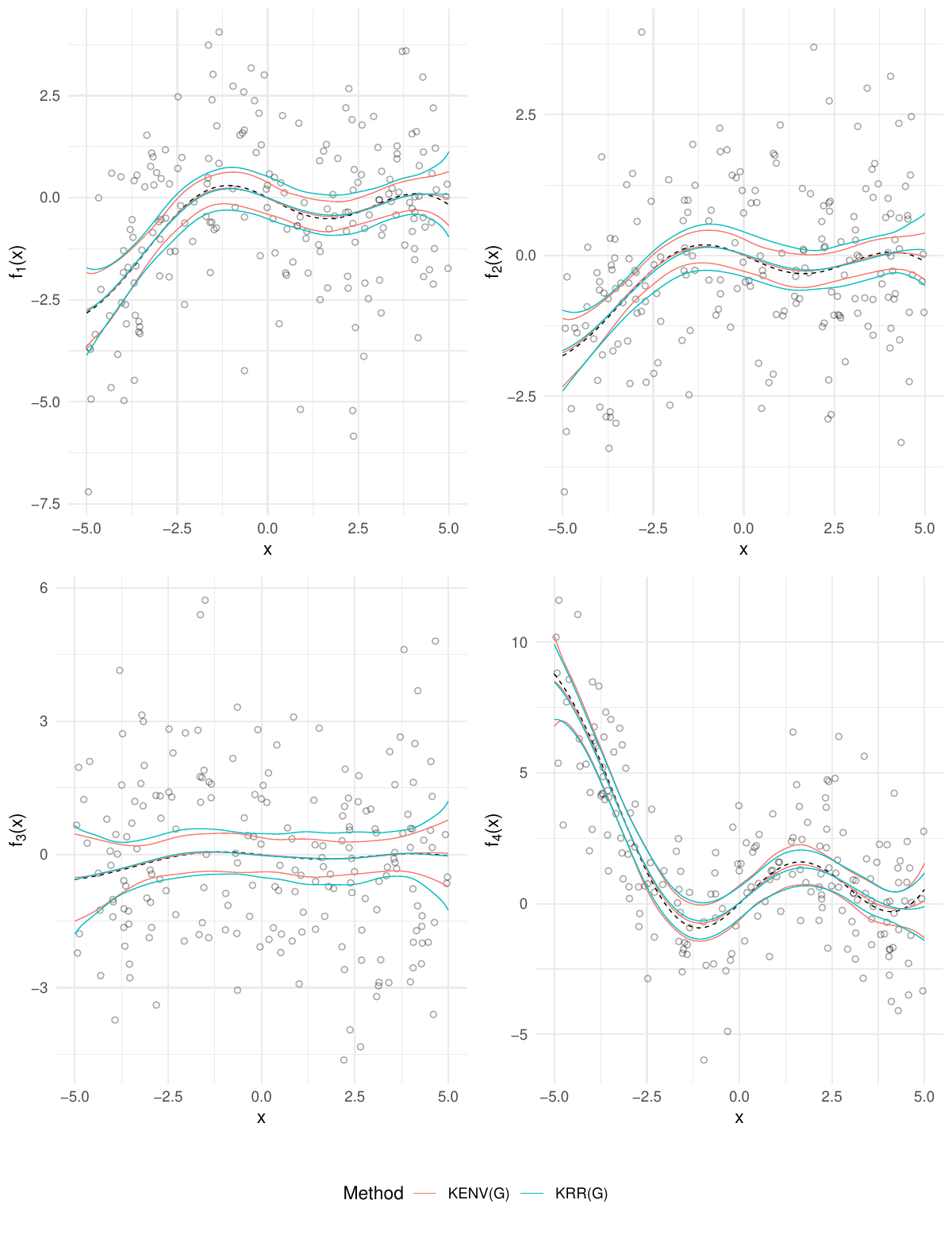}
	\end{figure}
		
	\section{Analysis of Chemometrics Data} \label{sec:real data}

	We further examine the empirical performance of KENV with a prediction comparison on the Chemometrics data of \citet{skagerberg_multivariate_1992}.
	The data were generated from a simulation of a reactor producing low-density polyethylene and contain $n = 56$ observations, $r = 6$ responses, and $p = 22$ predictors.
	The predictors include $20$ temperature measurements along the reactor, the wall temperature, and the solvent flow rate.
	The responses characterize the polymers produced by the reactor and include molecular weight, average molecular weight number, long chain branching,  short chain branching, the number of vinyl groups, and the number of vinylidene groups in the polymer chain.
	Finding that the responses were right skewed, we applied a log transformation and standardized them before fitting models.
	
	We consider the same methods as in Section \ref{sec:sims} for the comparison.
	We use $10$-fold CV to assess their prediction performance, computing the MSE and MAE for the predictions across the test folds and reporting their means.  
	We tune $u$, $\lambda$, and $\sigma$ using $10$-fold CV within each of the training folds.
	
	Table \ref{tbl:chemometrics data} reports the results from the comparison on the Chemometrics data.
	We see that KENV outperforms KRR with both the Gaussian and Laplacian kernels and that KENV(G) delivers far more accurate predictions than the SVM with a Gaussian kernel, demonstrating that the envelope structure of KENV leads to improved prediction performance in this real application.
	The chosen envelope dimension was $u = 3$ in $9$ of $10$ outer CV folds for KENV(G) and $8$ of $10$ folds for KENV(L), a substantial reduction from the original $6$ response components.
	Among the kernels, the Gaussian kernel outperforms the Laplacian kernel and the kernel methods dominate the linear methods, suggesting that the underlying relationship is not well-approximated by a linear model.
	Moreover, EENV delivers less accurate predictions than RR, again demonstrating its sensitivity to model misspecification and underscoring the need for nonparametric envelope methods.
	
	\spacingset{1.15}
	\begin{table}[t]
		\caption{Prediction comparison on Chemometrics data}
		\label{tbl:chemometrics data}
		\centering
		\begin{tabular}{llllllll}
			\toprule
			& KENV(G) & KRR(G) & KENV(L) & KRR(L) & EENV & RR & SVM\\
			\midrule
			\hspace{1em}MSE & \textbf{0.50} & 0.59 & 1.23 & 1.25 & 2.20 & 1.89 & 0.96\\
			\hspace{1em}MAE & \textbf{1.35} & 1.45 & 2.05 & 2.07 & 2.66 & 2.49 & 1.77
		\end{tabular}
	\end{table}
	\spacingset{1.7}
	
	\section{Discussion}\label{sec:discussion}
	In this paper, we have introduced the nonparametric response envelope model and the kernel-based KENV estimator for that model.
	Theoretically and empirically we have seen that KENV improves on the prediction performance of kernel ridge regression, decreasing the variance by estimating the material and immaterial parts of the response and using only the former to estimate the unknown regression function.
	At the same time, KENV extends existing response envelope methodology by allowing for a nonparametric relationship between the predictors and responses and by accommodating high-dimensional predictors.
	In empirical studies, we found that KENV consistently outperformed the linear enhanced response envelope estimator.
	In sum, KENV has several practical advantages: it is flexible, estimating a general nonparametric regression model; it directly estimates the regression functions and those estimates can be easily plotted and interpreted; it delivers superior prediction accuracy to alternative envelope and RKHS methods; it is able to handle both low and high-dimensional predictors; and it can be computed efficiently.
	
	There are a few natural directions for future work on nonparametric and kernel envelope methods.
	\citet{cook_reduced_rank_2015} proposed a hybrid reduced-rank and envelope regression model, which they showed to be at least as efficient as reduced-rank regression and response envelopes.
	We could similarly integrate KENV with kernel reduced-rank regression \citep{mukherjee_reduced_2011} to further improve on KENV.
	Drawing on \citet{allen_automatic_2013} and \citet{chen_double_2018}'s work on sparse kernel methods, we could also develop a sparse variant of KENV for selecting important features in high dimensions.
	Lastly, we could relax the assumptions in the theoretical study to allow for heteroscedasticity.

	\bibliographystyle{apalike}
	\bibliography{envlp}
	
	\newpage
	
	\appendix
	
	\setcounter{section}{0}
	\setcounter{equation}{0}

	\def\theequation{S\arabic{equation}}
	\def\thesection{S\arabic{section}}
	
	\section{Theoretical Proofs}
	
	\begin{proof}[Proof of Proposition \ref{prop:condition equivalence}]
		When \ref{condition a} holds, we see that $\bm Q_{\Eset} \fvec(\xvec) =\bm 0$ for all $\xvec \in \dom(\fvec)$. 
		As such, $\bm Q_{\Eset}\yvec = \bm Q_{\Eset} \fvec(\xvec) + \bm Q_{\Eset} \epvec = \bm Q_{\Eset} \epvec$, which implies \ref{condition i}. 
		Moving in the other direction, we note that $\fvec(\xvec) = \bm P_{\Eset}\fvec(\xvec) + \bm Q_{\Eset}\fvec(\xvec)$. Condition \ref{condition i} implies $\bm Q_{\Eset} \fvec(\xvec) = \bm 0$ for all $\xvec \in \dom(\fvec)$.
		As such, we have $\fvec(\xvec) = \bm P_{\Eset}\fvec(\xvec)$, which implies \ref{condition a}.
		The equivalence between \ref{condition ii} and \ref{condition b} follows immediately from the properties of the multivariate normal distribution.
	\end{proof}
	
	\begin{proof}[Proof of Theorem \ref{thm:rep theorem}]
		Suppose $\gvec = (g_1, \ldots, g_u)^T$ is the maximizer of $L_2(\Eset, \gvec, \Omat) - \frac{\lambda}{2} \rho(\gvec, \Omat)$.
		Define the subspace $ \mathcal{S}_K = \spann\{ K(\cdot, \xvec_i): i \in [n] \}$.
		We can express each $g_j$ as $g_j = g_j^K + g_j^{\perp}$ where $g_j^K$ is the projection of $g_j$ onto $\mathcal{S}_K$ and $g_j^{\perp}$ is orthogonal to $\mathcal{S}_K$.
		By the reproducing property of $K$, we see that for all $i \in [n], j \in [u]$
		$ g_j(x_i) = \inprod{g_j}{ K(\cdot, x_i)}_{\RKHS} =\inprod{g_j^K + g_j^{\perp}}{ K(\cdot, x_i)}_{\RKHS} = g_j^K(x_i) $.
		By extension, we have that $L_2(\Eset, \gvec, \Omat) = L_2(\Eset, \gvec^K, \Omat)$.
		
		Moving to the penalty, we find
		$\rho(\gvec, \Omat) 
		= \sum_{i=1}^u \norm{ [\Omat^{-1/2}]_{i\cdot} \gvec  }_{\RKHS}^2
		= \sum_{i=1}^u \norm{ [\Omat^{-1/2}]_{i\cdot} (\gvec^K + \gvec^{\perp})  }_{\RKHS}^2
		= \sum_{i=1}^u \left(\norm{ [\Omat^{-1/2}]_{i\cdot} \gvec^K  }_{\RKHS}^2 + \norm{ [\Omat^{-1/2}]_{i\cdot}  \gvec^{\perp}  }_{\RKHS}^2 + 2\inprod{[\Omat^{-1/2}]_{i\cdot}\gvec^K}{[\Omat^{-1/2}]_{i\cdot}\gvec^{\perp}}_{\RKHS}\right) $.
		Since $\inprod{g_i^K}{g_j^{\perp}}_{\RKHS} = 0$ for all $i,j \in [u]$, it is straightforward to show $\inprod{[\Omat^{-1/2}]_{i\cdot}\gvec^K}{[\Omat^{-1/2}]_{i\cdot}\gvec^{\perp}}_{\RKHS} = 0$ for all $i\in [u]$ .
		Therefore $\rho(\gvec, \Omat) %
		= \rho(\gvec^K, \Omat) + \rho(\gvec^{\perp}, \Omat)$.
		
		Combining these results, we see that 
		$L_2(\Eset, \gvec, \Omat) - \frac{\lambda}{2} \rho(\gvec, \Omat) \leq L_2(\Eset, \gvec^K, \Omat) - \frac{\lambda}{2} \rho(\gvec^K, \Omat)$.
		As such, the $\gvec^k$ must be the maximizer, meaning that the solution satisfies $g_j \in \mathcal S_{K}$ for all $j \in [u]$.
	\end{proof}
	
	\begin{proof}[Proof of Proposition \ref{prop:in sample risk comparison}]
		The KRR estimator is given by
		$
		\hat \fvec^R(\Xmat) = \bm K(\bm K + \lambda \bm I)^{-1}\Ymat
		$.
		When $\Gmat$ is known, the KENV estimator is
		$
		\hat \fvec^E_{\Gmat}(\Xmat) = \hat \fvec^R(\Xmat) \Gmat \Gmat^T,
		$
		which we recognize as the projection of the KRR estimator onto $\envlp$.
		
		We can express \eqref{eqn:regression model} in matrix form as
		$ \bm Y = \fvec(\Xmat) + \bm E $
		where $\bm E = [\epvec_1 \cdots \epvec_n]^T$, with $\E [\bm E] = \bm 0$ and $\Var[ \vect(\epvec) | \Xmat ] = \bm \Sigma \otimes \bm I_r$.
		Under \eqref{eqn:regression model}, we find that the bias of the KRR estimator is
		\begin{align*}
			\E[ \hat \fvec^R(\Xmat) | \Xmat ] - \fvec(\Xmat)
			& = \bm K(\bm K + \lambda \bm I)^{-1}\fvec(\Xmat) - \fvec(\Xmat) \\
			& = [\bm K(\bm K + \lambda \bm I)^{-1} - \bm I] \fvec(\Xmat)\\
			& = [\bm K(\bm K + \lambda \bm I)^{-1} - (\bm K + \lambda \bm I)(\bm K + \lambda \bm I)^{-1}] \fvec(\Xmat)\\
			& = -\lambda (\bm K + \lambda \bm I)^{-1} \fvec(\Xmat).
		\end{align*} 
		Under \eqref{eqn:envelope model}, we have that $\fvec(\Xmat) \in \envlp$. This implies that $\fvec(\Xmat)\Gmat \Gmat^T = \fvec(\Xmat)$.
		Using this fact, it is straightforward to show that KENV also has bias 
		$\E[ \hat \fvec^E_{\Gmat}(\Xmat) | \Xmat ] - \fvec(\Xmat) = -\lambda (\bm K + \lambda \bm I)^{-1} \fvec(\Xmat)$.
		
		We can express the KRR estimator's variance as
		\begin{align*}
			\Var[  \vect (\hat \fvec^R (\Xmat)) | \Xmat ]
			& = \Var[  \vect (\bm K(\bm K + \lambda \bm I)^{-1}\bm E) | \Xmat ] \\
			& = \Var [ (\bm I_r \otimes \bm K(\bm K + \lambda \bm I)^{-1} ) \vect(\bm E) | \Xmat  ]\\
			& = (\bm I_r \otimes \bm K(\bm K + \lambda \bm I)^{-1} ) (\bm \Sigma \otimes \bm I_r) (\bm I_r \otimes \bm K(\bm K + \lambda \bm I)^{-1} )^T\\
			& = \bm \Sigma \otimes \bm K(\bm K + \lambda \bm I)^{-2} \bm K.
		\end{align*} 
		As such, the variance term from \eqref{eqn:bv decomp} for the KRR estimator is
		\begin{align*}
			\tr\{ \Var[ \vect( \hat \fvec^R(\Xmat)) |\Xmat ] \} 
			& = \tr\{ \bm \Sigma \} \tr\{ \bm K(\bm K + \lambda \bm I)^{-2} \bm K\}\\
			& = \tr\{  \Gmat \Omat \Gmat^T + \Gmat_0 \Omat_0 \Gmat_0^T \}\tr\{ \bm K(\bm K + \lambda \bm I)^{-2} \bm K\}\\
			& = ( \tr\{ \Omat \} + \tr\{ \Omat_0 \} )\tr\{ \bm K(\bm K + \lambda \bm I)^{-2} \bm K\}
		\end{align*}
		where the last line follows from the linearity and cyclic properties of the trace.
		
		A similar derivation for the KENV estimator yields
		\begin{align*}
			\Var[  \vect (\hat \fvec^E_{\Gmat} (\Xmat)) | \Xmat ]
			& = \Var[  \vect (\bm K(\bm K + \lambda \bm I)^{-1}\bm E \Gmat \Gmat^T) | \Xmat ] \\
			& = \Var [ (\Gmat \Gmat^T \otimes \bm K(\bm K + \lambda \bm I)^{-1} ) \vect(\bm E) | \Xmat  ]\\
			& = \Gmat \Gmat^T \bm \Sigma \Gmat \Gmat^T \otimes \bm K(\bm K + \lambda \bm I)^{-2} \bm K\\
			& = \Gmat \Gmat^T (\Gmat \Omat \Gmat^T + \Gmat_0 \Omat_0 \Gmat_0^T ) \Gmat \Gmat^T \otimes \bm K(\bm K + \lambda \bm I)^{-2} \bm K \\
			& = \Gmat \Omat \Gmat^T \otimes \bm K(\bm K + \lambda \bm I)^{-2} \bm K.
		\end{align*} 
		Consequently, the variance term for the KENV estimator is
		$$
		\tr\{ \Var[ \vect( \hat \fvec^E_{\Gmat}(\Xmat)) |\Xmat ] \}  
		= \tr\{  \Gmat \Omat \Gmat^T \}\tr\{ \bm K(\bm K + \lambda \bm I)^{-2} \bm K\}
		= \tr\{  \Omat \}\tr\{ \bm K(\bm K + \lambda \bm I)^{-2} \bm K\}.
		$$	
		As such, we have the following relation between the in-sample risks for KENV and KRR
		$$
		R( \hat \fvec^E_{\Gmat}(\Xmat) | \Xmat ) 
		= R( \hat \fvec^R (\Xmat) | \Xmat ) - \tr\{  \Omat_0 \}\tr\{ \bm K(\bm K + \lambda \bm I)^{-2} \bm K\} 
		\leq R( \hat \fvec^R (\Xmat) | \Xmat )
		$$
		where the inequality is strict if $\tr\{  \Omat_0 \} > 0$.		
	\end{proof}
	
	\begin{proof}[Proof of Theorem \ref{thm:risk consistency}]
		Because $\Gmat$ and $\Omat$ are known, we can rewrite the KENV estimation problem as
		\begin{align}
			\argmin_{\gvec} \frac{1}{2}  \sum_{i=1}^n \norm{ \Omat^{-1/2}\Gmat^T\yvec_i - \Omat^{-1/2}\gvec(\xvec_i) }_2^2 + \frac{\lambda_n}{2} \rho (\Omat^{-1/2} \gvec) \label{eqn:KENV objective simplified}.
		\end{align}
		We define the standardized reduced responses to be $\bm z_i = \Omat^{-1/2}\Gmat^T\yvec_i$ for $i \in [n]$ and define the transformed regression function to be $\hvec = \Omat^{-1/2}\gvec$.
		Because $\Omat$ is positive definite, it is non-singular and the mapping $\hvec = \Omat^{-1/2}\gvec$ is bijective.
		Moreover, since $g_j \in \RKHS$ for $j \in [u]$, it also holds that  $h_j \in \RKHS$ for $j \in [u]$ since $\RKHS$ is closed under linear combinations.
		Using this new notation, we can further reduce \eqref{eqn:KENV objective simplified} to
		\begin{equation}
			\argmin_{\hvec} \sum_{j = 1}^u \left\{  \frac{1}{2}  \sum_{i=1}^n (z_{ij} - \hvec_j(\xvec_i) )^2 + \frac{\lambda_n}{2} \norm{h_j}^2_{\RKHS} \right\}, \label{eqn:KENV in terms of h}
		\end{equation}
		with a one-to-one correspondence between the solutions of \eqref{eqn:KENV objective simplified} and \eqref{eqn:KENV in terms of h}.
		We can solve \eqref{eqn:KENV in terms of h} via $u$ separate minimizations
		\begin{equation}
			\argmin_{h_j} \frac{1}{2}  \sum_{i=1}^n (z_{ij} - h_j(\xvec_i) )^2 + \frac{\lambda_n}{2} \norm{h_j}^2_{\RKHS} \text { \; for \; } j  \in [u]. \label{eqn:u separate minimizations}
		\end{equation}
		
		Our strategy will be to establish the risk consistency of the solutions to the $u$ problems in \eqref{eqn:u separate minimizations}, then combine these results to show that the solution to \eqref{eqn:KENV in terms of h} is also risk consistent.
		For ease of presentation, suppose $(\bm z, \xvec)$ come from the same distribution as $\{(\bm z_i, \xvec_i)\}_{i=1}^n$.
		We see that $\E [z_j^2] < \infty$ for $j \in [u]$ since $\Gmat$ and $\Omat$ are fixed and $\E [y_j^2] < \infty$ for $j \in [r]$ .
		The risk for the $j$th problem is 
		$
		R_j(h_j) = \E[ (z_j - h_j(\xvec))^2 ]
		$
		and the optimal risk is 
		$
		R_j^* = \inf_{h_j} R_j(h_j)
		$.
		Let $\hat h_j$ denote the solution to the $j$th problem in \eqref{eqn:u separate minimizations}.
		Theorem 12 of \cite{christmann_consistency_2007} establishes that $R_j(\hat h_j) \to _p R_j^*$ as $n \to \infty$ for $j \in [u]$ if $\lambda_n/n \to 0$ and $\lambda_n/n^{3/4} \to \infty$. %
		
		Let $\tau > 0$. From our $u$ risk consistency results, we know that $\forall j \in [u]$ $\exists N_j > 0$ such that
		$
		P(R_j(h_j) - R_j^* > \tau/u) < \tau/u \text{ for all } n > N_j
		$
		(note that $R_j(h_j) > R_j^*$ by definition, so we drop the absolute value).
		We see that the risk for the multivariate problem is
		$
		R(\hvec) = \E[ \sum_{j = 1}^{u} (z_j - h_j(\xvec))^2  ] = \sum_{j = 1}^{u} R_j(h_j)
		$
		and that the optimal risk is $R^* = \sum_{j = 1}^{u} R_j^*$.
		Define $N = \max_{j \in [u]} N_j$.
		Applying the union bound, we find that
		\begin{align*}
			P(R(\hat \hvec) - R^* > \tau  )
			& = P\left(\sum_{j = 1}^{u} \left\{R_j(h_j) - R_j^*\right\} > \tau  \right) 
			\leq P\left(\bigcup_{j = 1}^{u} \left\{R_j(h_j) - R_j^*\right\} > \tau/u \right) \\
			& \leq \sum_{j = 1}^{u} P\left( R_j(h_j) - R_j^* > \tau/u \right)
			< \tau
		\end{align*}
		for all $n > N$.
		Therefore $ R(\hat \hvec) \to_p R^*$ as $n \to \infty$.
	\end{proof}
	
\end{document}